\documentclass{article}
\usepackage{ijcai23}

\usepackage{times}
\usepackage{soul}
\usepackage{url}
\usepackage[hidelinks]{hyperref}
\usepackage[utf8]{inputenc}
\usepackage[small]{caption}
\usepackage{graphicx}
\usepackage{amsmath}
\usepackage{amsthm}
\usepackage{booktabs}
\usepackage{algorithm}
\usepackage{algorithmic}
\usepackage[switch]{lineno}
\usepackage{balance} 
\usepackage{graphicx}
\usepackage{times}

\usepackage{soul}
\usepackage{url}
\usepackage[hidelinks]{hyperref}
\usepackage[utf8]{inputenc}
\usepackage{caption}
\usepackage{amsmath}
\usepackage{appendix}
\usepackage{booktabs}
\usepackage{array}

\usepackage{amsfonts}
\usepackage{amsmath}
\usepackage{framed}
\usepackage{subfig}

\usepackage{algorithm}
\usepackage{algorithmic}
\usepackage{float}  
\usepackage{lipsum}
\makeatletter
\newenvironment{breakablealgorithm}
{
	\begin{center}
		\refstepcounter{algorithm}
		\hrule height.8pt depth0pt \kern2pt
		\renewcommand{\caption}[2][\relax]{
			{\raggedright\textbf{\ALG@name~\thealgorithm} ##2\par}%
			\ifx\relax##1\relax 
			\addcontentsline{loa}{algorithm}{\protect\numberline{\thealgorithm}##2}%
			\else 
			\addcontentsline{loa}{algorithm}{\protect\numberline{\thealgorithm}##1}%
			\fi
			\kern2pt\hrule\kern2pt
		}
	}{
		\kern2pt\hrule\relax
	\end{center}
}
\makeatother

\newtheorem{theorem}{Theorem}
\newtheorem{definition}{Definition}
\newtheorem{lemma}{Lemma}
\newtheorem{corollary}{Corollary}
\newtheorem{proposition}{Proposition}
\newtheorem{observation}{Observation}

\title{Diffusion Multi-unit Auctions with Diminishing Marginal Utility Buyers}

\author{
Haolin Liu$^{1*}$
\and
Xinyuan Lian$^{2}$\footnote{Equal Contribution} \And
Dengji Zhao$^{2}$
\affiliations
$^1$University of Virginia\\
$^2$ShanghaiTech University
\emails
srs8rh@virginia.edu,
\{lianxy, zhaodj\}@shanghaitech.edu.cn
}

\begin{document}

\maketitle


\begin{abstract}
We consider an auction design problem where a seller sells multiple homogeneous items to a set of connected buyers. Each buyer only knows the buyers she directly connects with and has a diminishing marginal utility valuation for the items. The seller initially only connects to some buyers who can be directly invited to the sale by the seller. Our goal is to design an auction to incentivize the buyers who are aware of the auction to further invite their neighbors to join the auction. This is challenging because the buyers are competing for the items and they would not invite each other by default. Thus, rewards need to be given to buyers who diffuse information, but the rewards should be carefully designed to guarantee both invitation incentives and the seller's revenue. Solutions have been proposed recently for the settings where each buyer requires at most one unit and demonstrated the difficulties of the design. We move this forward to propose the very first diffusion auction for the multi-unit demand settings to improve both the social welfare and the seller’s revenue.
\end{abstract}



\section{Introduction}
Multi-unit auctions refer to auctions where multiple homogeneous items are available for sale in a single auction. They are commonly used in many real-world markets to allocate scarce resources such as emissions permits, electricity and spectrum licenses~\cite{hortaccsu2008understanding,huang2006auction}. A common feature of these markets is that the participants are known in advance, and the seller can hold the classic Vickrey-Clarke-Groves (VCG) mechanism to get a good outcome~\cite{vickrey1961counterspeculation,clarke1971multipart,groves1973incentives}. In order to further improve social welfare and revenue, an intuitive method is to advertise the sale to involve more buyers. \citeauthor{bulow1994auctions}~\shortcite{bulow1994auctions} showed that 
the VCG mechanism among $n+1$ buyers for selling one item will give more revenue in expectation than Myerson's optimal auction among $n$ buyers~\cite{myerson1981optimal}. Thus it is worthwhile for the seller to expand the market through advertisements.
However, advertisement does not come for free and if the advertisements can't attract enough valuable buyers, the seller's total utility may decrease.

In order to attract more buyers without sacrificing the seller's utility, diffusion mechanism design has been proposed in recent years~\cite{zhao2021mechanism,li2022diffusion}. Diffusion mechanisms utilize the participants' connections to attract more buyers. This is done by incentivizing each participant who is aware of the market to invite all her neighbors to join the market. This is challenging because the participants are competing for the same resources. Therefore, proper rewards are designed for buyers who invite valuable buyers, but at the same time, we need to guarantee the seller is incentivized to use the mechanism. 

The very first diffusion mechanism, called the information diffusion mechanism (IDM), was proposed in \citeauthor{li2017mechanism}~\shortcite{li2017mechanism}, where a simple case of a seller selling one item via a social network is studied. 
Following this work, two studies tried to extend the model to more general settings. \citeauthor{zhao2019selling}~\shortcite{zhao2019selling} first tried to generalize the setting from single-unit-supply to multi-unit-supply case and proposed a mechanism called generalized IDM (GIDM). However, an error in one key proof was identified later~\cite{takanashi2019efficiency}. Then, another mechanism, called DNA-MU, was proposed for the same setting~\cite{kawasaki2020strategy}. Surprisingly, the new mechanism is also problematic (we will show a counter-example later). Therefore, this problem is still open even after a few hard tries. This also indicates that generalizing IDM to more general settings is very difficult (see~\cite{zhao22ijcai} for a detailed discussion).


In this paper, we keep working on this open question. We further extend the setting to a multi-unit-supply and multi-unit-demand case, which is more challenging (Section $3$). That is, the seller sells multiple units and each buyer can buy multiple units with a diminishing marginal utility function (i.e., the valuation for receiving more units is non-increasing). This setting was briefly touched by \citeauthor{takanashi2019efficiency}~\shortcite{takanashi2019efficiency}. However, their solution restricts that each winner can only get $n_k$ units where $n_k$ is predefined, and the remaining units are discarded.  
Such a design degenerates multi-unit demand into unit demand. Also, their solution may give the seller negative revenue. Here, we propose a new solution without restricting the allocation and making sure the seller's revenue is improved.

Except for extending the studies to more general settings, single-unit diffusion auctions have also been extensively studied. Following IDM, a general class of single-item diffusion auctions was proposed by \citeauthor{li2019diffusion}~\shortcite{li2019diffusion}. To make the reward distribution fairer, a fair diffusion mechanism was proposed~\cite{ZhangZC20}. For the circumstance where the seller's aim is not for profit, a redistribution diffusion mechanism was studied~\cite{zhang2019incentivize}. That is, there are many variations for the single-unit setting and then \citeauthor{li2020incentive}~\shortcite{li2020incentive} further characterized the conditions to achieve the invitation incentive. 
Finally, similar invitation incentives are studied in other settings such as matching and cooperative games~\cite{zhao2021mechanism}.

\section{The Model}
We consider a setting where a seller $s$ sells $\mathcal{K}\ge 1$ homogeneous items via a network. In addition to the seller, the social network consists of $n$ potential buyers denoted by $N = \{1, \cdots, n\}$.  Each $i \in N$ has a private marginal decreasing utility function for the $\mathcal{K}$ items which is denoted by a value vector $v_i = (v_i^1, \cdots, v_i^{\mathcal{K}})$ where $v_i^1 \ge v_i^2 \ge \cdots \ge v_i^{\mathcal{K}} \ge 0$. Then $i$'s valuation for receiving $m \ge 1$ units is denoted by $v_i(m) = \sum_{k = 1}^{m} v_i^{k}$, and the valuation for receiving nothing is $v_i(0) = 0$. Each buyer $i\in N$ has a set of neighbors denoted by $r_i \subseteq N \cup \{s\}$ and she does not know the existence of the others except for $r_i$. Thus, the seller is also only aware of her neighbors initially. Let the neighbors of the seller be $r_s$.

Let $\theta_i = (v_i, r_i)$ be the \emph{type} of buyer $i\in N$ and $\theta = (\theta_1,\cdots,\theta_n)$ be the type profile of all buyers. $\theta$ can also be written as $(\theta_i, \theta_{-i})$ where $\theta_{-i}$ is the type profile of all buyers except for $i$. Let $\Theta_i$ be the type space of buyer $i$ and $\Theta$ be the type profile space of all buyers. Since the seller initially only connects to a few buyers, we want to design auction mechanisms that ask each buyer to report her valuation on the items and invite her neighbors to join the mechanism. This is mathematically modeled by reporting her type. Let $\hat{\theta}_i = (\hat{v}_i, \hat{r}_i)$ be buyer $i$'s type report where $\hat{r}_i \subseteq r_i$ because $i$ can not invite someone she does not know. Let $\hat{\theta} = (\hat{\theta}_1,\cdots, \hat{\theta}_n)$ be the report profile of all buyers. 

A general auction mechanism consists of an \emph{allocation policy} $\pi = (\pi_i)_{i\in N}$ and a \emph{payment policy} $p = (p_i)_{i\in N}$. Given a report profile $\hat{\theta}$, $\pi_i(\hat{\theta})\in \{0, 1, \cdots, \mathcal{K}\}$ is the number of items $i$ receives and $\sum_{i \in N } \pi_i(\hat{\theta}) \le \mathcal{K}$. $p_i(\hat{\theta}) \in \mathbb{R}$ is the payment that $i$ pays to the mechanism. If $p_i(\hat{\theta}) < 0$, then $i$ receives $|p_i(\hat{\theta})|$ from the mechanism. 

Since we assume that each participant is only aware of her neighbors,  initially only the seller's neighbors are invited to join the auction. Other buyers who are not properly invited by early joined buyers cannot join the auction, i.e., their reports cannot be used by the mechanism. Therefore, we have some additional constraints for the mechanism, which will be defined as the diffusion auction mechanism. 

\begin{definition}
	Given a report profile $\hat{\theta}$, an \emph{invitation chain} from the seller $s$ to buyer $i$ is a sequence of $(s,j_1,\cdots, j_l,j_{l+1}, \cdots, j_m,i)$ such that $j_1\in r_s$ and for all $1<l \leq m$, $j_l \in \hat{r}_{j_{l-1}}$, $i\in \hat{r}_{j_m}$ and no buyer appears twice in the sequence. If there is an invitation chain from the seller $s$ to buyer $i$, then we say buyer $i$ is \emph{valid} in the auction. Let $Q(\hat{\theta})$ be the set of all valid buyers under $\hat{\theta}$.
\end{definition}
Let $d_i(\hat{\theta})$ be the shortest length of all the invitation chains from seller $s$ to $i$ for each buyer $i \in Q(\hat{\theta})$. We denote $\mathcal{L}_d(\hat{\theta})$ the set of valid buyers whose shortest length of invitation chains is $d$, i.e., $\mathcal{L}_d(\hat{\theta}) = \{i | i\in Q(\hat{\theta}), d_i(\hat{\theta}) = d \}$. We also call $\mathcal{L}_d(\hat{\theta})$ layer $d$. Let $\mathcal{L}_{< l}(\hat{\theta}) = \bigcup_{1 \le i < l} \mathcal{L}_{i}(\hat{\theta})$ and $\mathcal{L}_{> l}(\hat{\theta}) = \bigcup_{i > l} \mathcal{L}_{i}(\hat{\theta})$

\begin{definition}
	A diffusion auction mechanism $(\pi,p)$ is an auction mechanism, where for all $\hat{\theta}$:
	\begin{itemize}
	\item for all invalid buyers $i \notin Q(\hat{\theta})$, $\pi_i(\hat{\theta}) = 0$ and $p_i(\hat{\theta}) = 0$.
	\item for all valid buyers $i \in Q(\hat{\theta})$, $\pi_i(\hat{\theta})$ and $p_i(\hat{\theta})$ are independent of the reports of buyers not in $Q(\hat{\theta})$.
	\end{itemize}
\end{definition}

Given a buyer $i$ of type $\theta_i = (v_i, r_i)$ and a report profile $\hat{\theta}$, the \emph{utility} of $i$ under a diffusion auction mechanism $(\pi, p)$ is defined as $u_i(\theta_i, \hat{\theta}, (\pi, p)) = v_i(\pi_i(\hat{\theta}))  - p_i(\hat{\theta})$.
For simplicity, we will use $u_i(\hat{\theta})$ to represent $u_i(\theta_i, \hat{\theta}, (\pi, p))$ when the mechanism is clear.

We say a diffusion auction mechanism is individually rational (IR) if, for each buyer, her utility is non-negative when she truthfully reports her valuation, no matter how many neighbors she invites and what the others do. It means that buyers' invitation behaviour will not make their utility negative as long as she reports their valuation truthfully. That is, we do not force buyers to invite all their neighbors to be IR.

\begin{definition}
	A diffusion auction mechanism $(\pi,p)$ is \emph{individually rational} (IR) if $u_i((v_i, \hat{r}_i), \hat{\theta}_{-i}) \geq 0$ for all $i \in N$, all $\hat{r}_i \subseteq r_i$, and all $\hat{\theta}_{-i}$.
\end{definition}

We say a diffusion auction mechanism is incentive compatible (IC) if for each buyer, truthfully reporting her valuation and inviting all her neighbors (i.e. reporting type truthfully) is a dominant strategy.

\begin{definition}
	A diffusion auction mechanism $(\pi,p)$ is \emph{incentive compatible} (IC) if $u_i(\theta_i, \hat{\theta}_{-i}) \geq u_i(\hat{\theta}_i, \hat{\theta}_{-i})$ for all $i\in N$, all $\hat{\theta}_i$ and all $\hat{\theta}_{-i}$.
\end{definition}

We say an auction is non-wasteful if all items can be allocated at the end of the auction. Non-wastefulness makes sure no item is discarded in auctions, which is important for social welfare. A similar definition is also given in \citeauthor{kawasaki2020strategy}~\shortcite{kawasaki2020strategy}.
\begin{definition}
    Given $N \neq \emptyset$, an auction mechanism is non-wasteful if $\sum_{i \in N}\pi_i(\hat{\theta}) = \mathcal{K}$ for any $\hat{\theta}$.
\end{definition}

\section{The Existing Mechanisms}
\label{faliure}
In this section, we will briefly discuss the techniques of the existing mechanisms and compare them with our design.


\subsection{The Failure of DNA-MU}
As we mentioned in the introduction, for the multi-unit-supply and single-unit-demand setting, two mechanisms have been proposed, GIDM~\cite{zhao2019selling} and DNA-MU~\cite{kawasaki2020strategy}. The problem of GIDM was identified early~\cite{takanashi2019efficiency}. Here, we show that DNA-MU also fails in the same counter-example.  

The description of DNA-MU on a tree is given in Algorithm \ref{DNA}, which is not hard to follow. The mechanism simply allocates the units to the buyers closer to the seller until no more units are left. The allocation for each chosen buyer is done by removing all the buyer's descendants (subtree) and checking whether the buyer's valuation is among the top $\mathcal{K}'$ highest of the remaining buyers ($\mathcal{K}'$ is the number of remaining units). Some technical terms: given a report profile $\hat{\theta}$, a subset $S \subseteq Q(\hat{\theta})$, and an integer $\mathcal{K}' \leq \mathcal{K}$, let $v^* (S, \mathcal{K}')$ denote the $\mathcal{K}'$-th highest value in buyer set $S$. 
If $|S| < \mathcal{K}'$, then $v^* (S, \mathcal{K}')=0$. 
Let $l^{max}$ be the number of layers of the tree network and $E_i(\hat{\theta})$ be the set of descendants of buyer $i$. 

\renewcommand{\algorithmicrequire}{\textbf{Input:}}
\renewcommand{\algorithmicensure}{\textbf{Output:}}

\begin{breakablealgorithm}
	\caption{DNA-MU on a Tree}
	\begin{algorithmic}[1]
		\REQUIRE A report profile $\hat{\theta}$;
		\ENSURE  $\pi(\hat{\theta})$ and $p(\hat{\theta})$;
		
		\STATE Construct the network according to $\hat{\theta}$;
		\STATE Initialize $\mathcal{K}' = \mathcal{K}$ and $W=\emptyset$;
		\FOR{$l = 1, 2 , \cdots, l^{max}$}
		    \FOR {$i \in \mathcal{L}_l (\hat{\theta})$ with random order} 
		        \STATE $p^*_i=v^*(Q(\hat{\theta})\setminus (E_i(\hat{\theta}) \cup W \cup \{i\}), \mathcal{K}')$ which is the price for buyer $i$;
        		\IF{$\hat{v}_i (\hat{\theta}) \geq p^*_i$}
        		\STATE $\pi_i(\hat{\theta})=1$, $p_i(\hat{\theta}) = p_i^*$, $\mathcal{K}' = \mathcal{K}' -1$, $W=W\cup \{i\}$;
        		\ELSE
        		\STATE $\pi_i(\hat{\theta})=0$, $p_i(\hat{\theta}) = 0$;
        		\ENDIF
    		\ENDFOR
		\ENDFOR
		\STATE Return $\pi_i(\hat{\theta})$ and $p_i(\hat{\theta})$ for each buyer $i$.
	\end{algorithmic}
	\label{DNA}
\end{breakablealgorithm}

The counter-example of the above mechanism is shown in Figure 1. We can run Algorithm \ref{DNA} on the network with $4$ units. If buyer B invites F, she cannot win the item and her utility is $0$ (the winners are \{B, C, D, G\}). However, if E does not invite F, she will get one unit with utility $1$ (the winners are \{A, B, C, E\}). The error in their proof comes from their Lemma 1~\cite{kawasaki2020strategy}. In this lemma, they claim that if buyer $i$ invites fewer neighbors, 
buyer $j$ before $i$ who is not the ancestor of $i$ may have a lower price to receive a unit, which decreases the chance of $i$ winning a unit. 
However, as we show in the example, this claim does not hold. In the example, when E stops inviting F, buyer A indeed gets one item with a lower price. However, the competition for the remaining buyers also increases. The winning prices for both E's ancestor D and non-ancestor G increase. Eventually, both D and G lose their units, and E receives one unit.

\begin{figure}[htbp]  
    \centering    
    \includegraphics[scale=0.35]{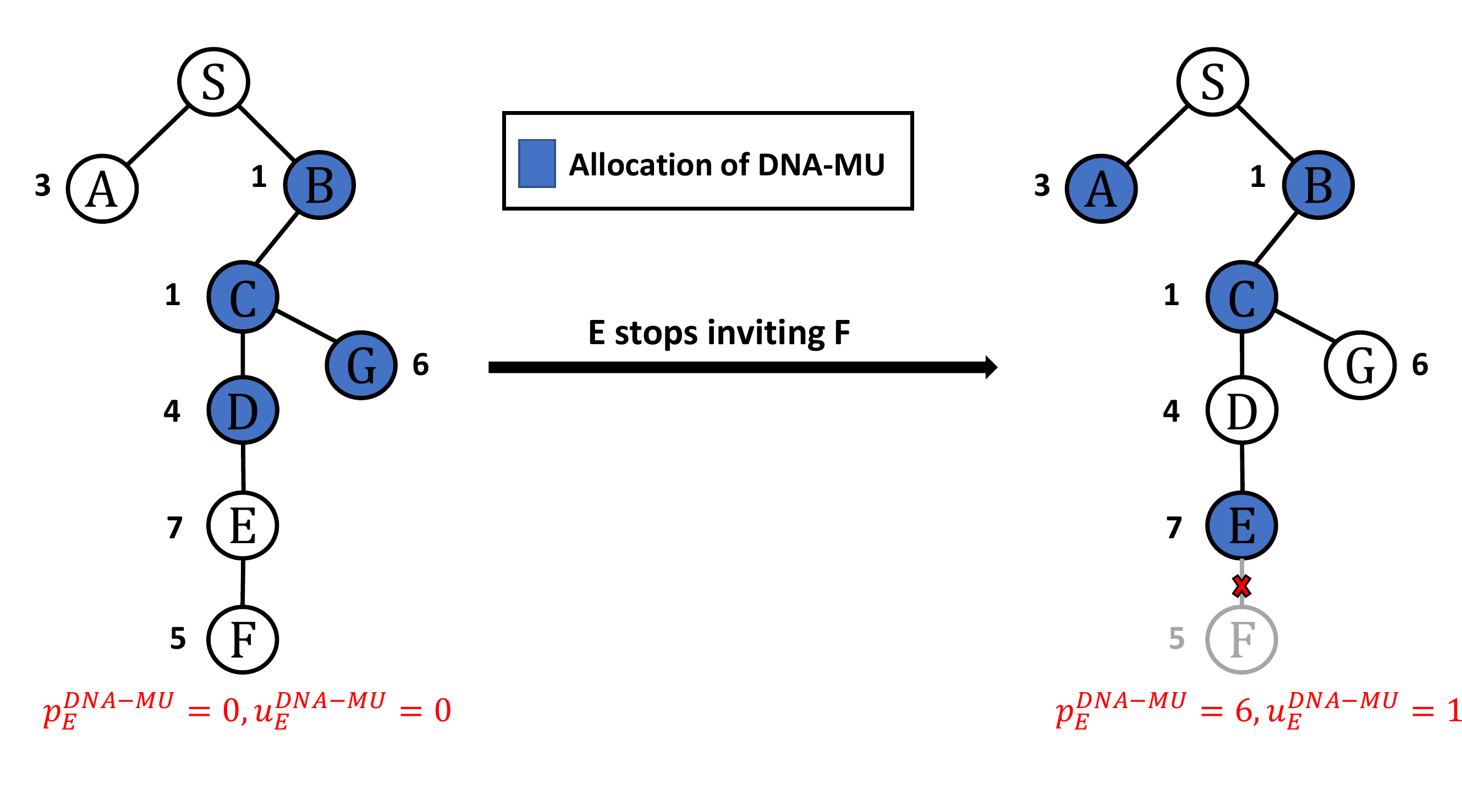}   
    \caption{A counter-example for DNA-MU}
    \label{figCompare}
\end{figure}

Interestingly, this counter-example is exactly the same counter-example for GIDM. Both GIDM and DNA-MU try to prove that a buyer $i$ can have a higher chance to get one item or gain a positive utility by inviting more valuable neighbors. However, their designs overlooked the impact of the buyers who are closer to the seller than $i$, and $i$ can create extra competition among them.

\subsection{The Existing Mechanisms v.s. Our Mechanism}
Although the two earlier extensions failed in some corner cases, their methodologies are very valuable. We will discuss and compare their methodologies with ours here. 

\subsubsection{Resale Incentives}
Both IDM and GIDM used the idea of resale to incentivize buyers to invite their neighbors~\cite{li2017mechanism,zhao2019selling}. The marginal contribution/utility of a buyer $i$ is computed roughly as follows:
\begin{itemize}
    \item Assume there are $k$ units finally allocated to buyer $i$ and her descendants (who will be disconnected from the seller without $i$).
    \item If $i$ hide all her descendants, she needs to pay $p_k$ to win the $k$ units alone.
    \item If $i$ resells the $k$ units among her descendants only with the same mechanism, she will receive $p_k'$.
    \item Then IDM/GIDM gives marginal utility $p_k' - p_k$ to $i$.
\end{itemize}

Another way to interpret the marginal utility of buyer $i$ in IDM/GIDM is the social welfare increase due to the participation of $i$'s descendants. IDM/GIDM gives all the marginal social welfare increase to $i$, so she is not afraid of inviting all her neighbors to join the mechanism.

\subsubsection{Winning Incentives}
The incentive behind DNA-MU (Algorithm $1$) is that a buyer can get a higher chance to win one unit by inviting more valuable neighbors~\cite{kawasaki2020strategy}. There is no resale benefit, i.e., the winning of a buyer's descendants will not give any reward to the buyer. Instead, large value descendants of a buyer $i$ will increase the winning prices of the non-ancestor buyers before $i$, and therefore, more units will be left and  $i$ may get a chance to receive one unit.

\subsubsection{Our Incentives}
Our mechanism is called the Layer-based Diffusion Mechanism (LDM) which is defined in Algorithm $2$. The incentive is a kind of combination of resale and winning. LDM prioritises buyers like DNA-MU according to their distances from the seller. It also utilizes the resale marginal utility to pay the buyers. The key differences are:
\begin{itemize}
    \item LDM not only solves the open incentive problem for unit demand but also works for multi-unit demand.
    \item DNA-MU determines the allocation for one buyer only at each step, while LDM decides the allocation for a layer in one step.
    \item In LDM, to decide the allocation of a layer, it has to remove all the potential competitors after the layer, because they have the incentive to manipulate this layer's allocation. This is also a kind of limitation of the mechanism. There might be another way to remove fewer buyers to determine their allocations, which will potentially increase social welfare and revenue.
\end{itemize}

\begin{figure}[t]  
    \centering    
    
    \subfloat[Unit-demand without D\label{figCompare1}] 
    {
        \centering         
        \includegraphics[scale=0.4]{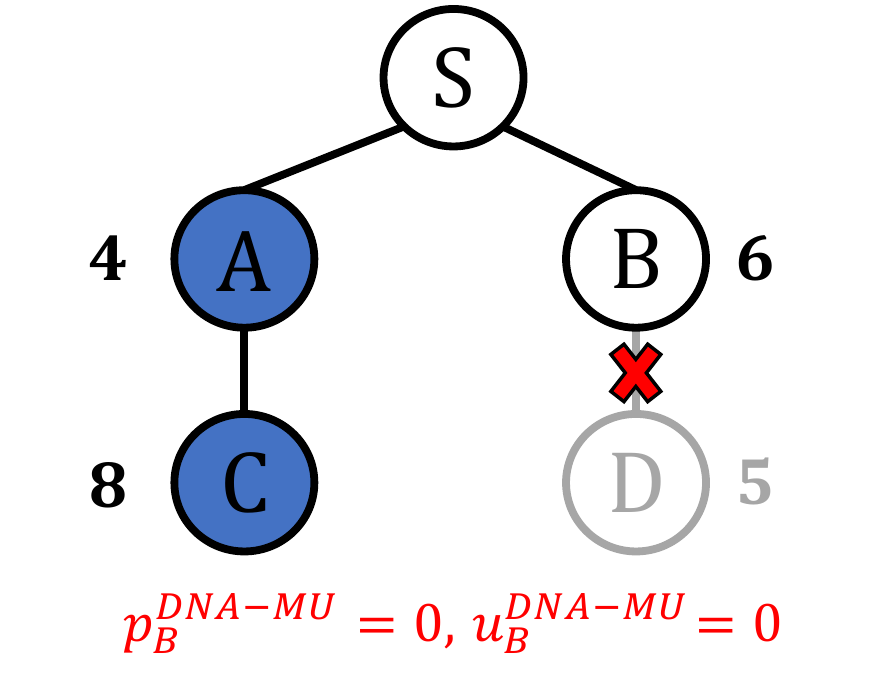}   
    }\hspace{3mm}
    \subfloat[Unit-demand with D\label{figCompare2}] 
    {
        \centering    
        \includegraphics[scale=0.4]{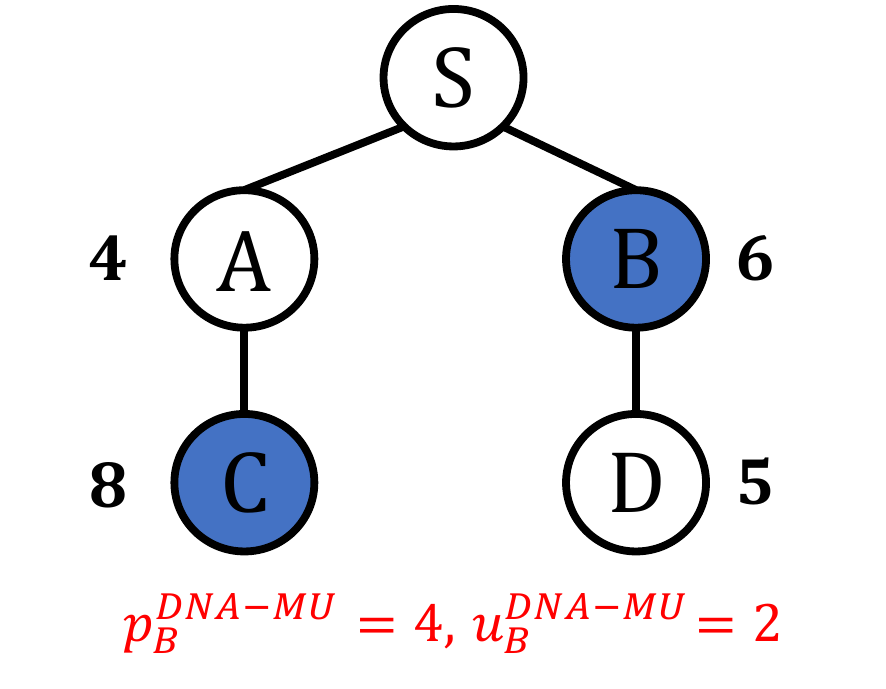}    
    }\hspace{3mm}
    \subfloat[Multi-unit-demand without D\label{figCompare3}] 
    {
        \centering    
        \includegraphics[scale=0.4]{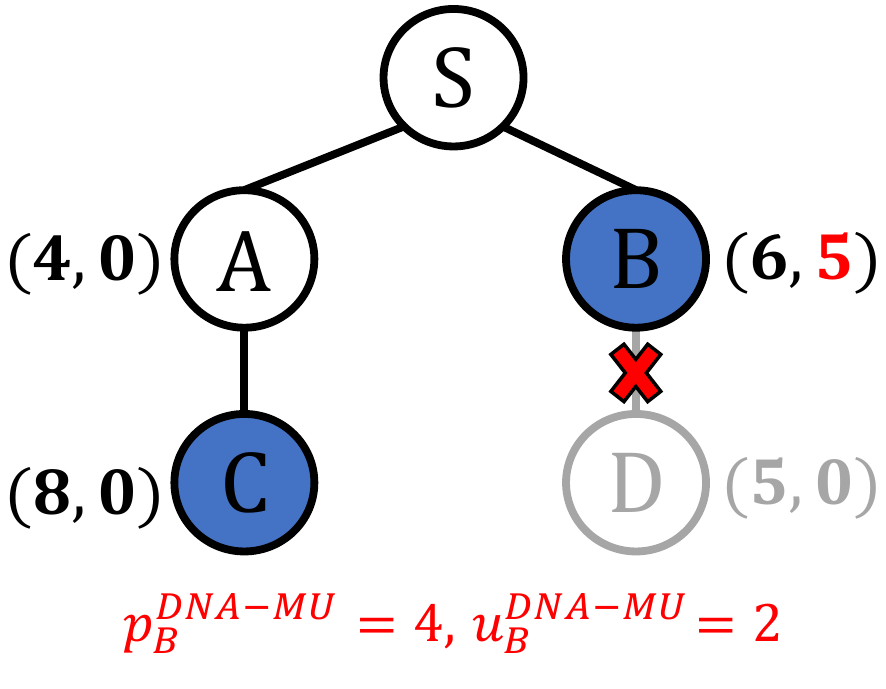}    
    }
    \caption{An example of misreporting in multi-unit-demand setting}
    \label{figCompare}
\end{figure}

\subsubsection{The Difficulty of Multi-unit Demand}
Despite the issues of GIDM and DNA-MU in the unit-demand setting, their methods will not work in multi-unit-demand settings.
The intuitive reason is that these mechanisms allow buyers to utilize their children's valuations to win items. Such a design is not applicable when buyers have a multi-dimensional valuation function, because a buyer can misreport her valuations to mimic the same demand as she has different neighbors. Figure~\ref{figCompare} shows an example of how this can be achieved in DNA-MU. Assume there are two items to sell and items are allocated to buyers with blue color. Figures~\ref{figCompare1} and \ref{figCompare2} show that in the single-unit-demand setting, buyer B can gain more utility from inviting buyer D. However, if we go to a multi-unit-demand setting with two units, then B can misreport her second valuation to $5$ to mimic that she has a child with a valuation $5$ to gain (Figure~\ref{figCompare3}). Thus, there is a conflict between reporting valuation truthfully and inviting neighbors. 
Therefore, the challenge seems a lot harder in the multi-unit-demand case.


\section{Diffusion Auction Mechanism on Trees}
In this section, we first design an IR and IC mechanism called layer-based diffusion mechanism on trees (LDM-Tree), for buyers with multi-dimensional valuations on tree-structure networks. Given a report profile $\hat{\theta}$, a tree $\mathcal{T}(\hat{\theta}) = (V(\hat{\theta}) \cup \{s\}, E(\hat{\theta}))$ can be constructed where seller $s$ is the root and $V(\hat{\theta}) = \{i | i \in Q(\hat{\theta})\}$. For each $i \in V(\hat{\theta}) \cup \{s\}$ and her neighbor $j \in \hat{r}_i$, there is an edge $(i,j) \in E(\hat{\theta})$. We use $C_i(\hat{\theta})$ to denote buyer $i$'s children in $\mathcal{T}(\hat{\theta})$ for any $i \in V(\hat{\theta})$. We use $l^{max}$ to denote the total number of layers in $\mathcal{T}(\hat{\theta})$. Generally speaking, LDM-Tree has the following three steps.

\begin{enumerate}
    \item \textbf{Prioritize:} LDM-Tree first prioritizes buyers by the layers. Layers closer to the seller have higher priority to be considered in the allocation. Then LDM traverses all layers based on priority with the following allocation and payment rules.
    
    \item \textbf{Allocation Policy:} When considering the allocation of layer $l$, we first remove a set of buyers $\mathcal{R}_l(\hat{\theta}) \subseteq \mathcal{L}_{> l}(\hat{\theta})$ from the higher layers that contains all potential competitors.(We will define exact $\mathcal{R}_l(\hat{\theta})$ later) Then in the remaining buyers, we fix the allocations to all layer $q < l$ and compute the constrained optimal social welfare $\mathcal{SW}_{-\mathcal{R}_l}(\hat{\theta})$ to determine the allocation of layer $l$.



    \item \textbf{Payment Policy:} When computing the payment of $i \in \mathcal{L}_l(\hat{\theta})$, we first remove $D_i(\hat{\theta}) =  \mathcal{R}_l(\hat{\theta}) \cup C_i(\hat{\theta}) \cup \{i\}$ which hide all $i$'s children. Then in the remaining buyers, we compute the maximal social welfare given that the allocations to lower layers are fixed, which is denoted by $\mathcal{SW}_{-D_i}(\hat{\theta})$. The social welfare difference $\mathcal{SW}_{-\mathcal{R}_l}(\hat{\theta}) - \mathcal{SW}_{-D_i}$ is used to calculate diffusion rewards(payments) which is a kind of resale revenue mentioned in Section \ref{faliure}.
    
\end{enumerate}

\noindent The formal definition of LDM-Tree is given in Algorithm \ref{NRM-T}.

\begin{breakablealgorithm}
	\caption{Layer-based Diffusion Mechanism on Trees}
	\begin{algorithmic}[1]
		\REQUIRE A report profile $\hat{\theta}$;
		\ENSURE  $\pi(\hat{\theta})$ and $p(\hat{\theta})$;
		
		\STATE Construct the tree $\mathcal{T}(\hat{\theta})$
		\STATE Initialize $\mathcal{K}^{remain} = \mathcal{K}$;
		\FOR{$l = 1, 2 , \cdots, l^{max}$}
		\STATE Compute the following constrained optimization problem and let $\pi^l(\hat{\theta})$ be the optimal solution.
        \begin{align*}
				\max\limits_{\pi(\hat{\theta})} \quad   & \mathcal{SW}_{-\mathcal{R}_{\emph{l}}}(\hat{\theta})  =\sum_{i \in Q(\hat{\theta}) \setminus \mathcal{R_{\emph{l}}}(\hat{\theta})} \hat{v}_i (\pi_i(\hat{\theta}))  \\
				\textrm{s.t.} \quad
				& \text{When $l \neq 1$, } \,\, \forall q<l, \,\, \forall j \in \mathcal{L_{\emph{q}}},  \,\, \pi_j(\hat{\theta}) = \pi_j^q(\hat{\theta}) 
		\end{align*}

		
		
		\FOR{$i \in \mathcal{L_\emph{l}}(\hat{\theta})$}		
		\STATE Compute the constrained optimization problem:
		 \begin{align*}
				\max\limits_{\pi(\hat{\theta})} \quad   & \mathcal{SW}_{-D_i}(\hat{\theta}) = \sum_{j\in Q(\hat{\theta}) \setminus D_i} \hat{v}_j (\pi_j(\hat{\theta}))  \\
				\textrm{s.t.} \quad
				& \text{When $l \neq 1$, } \,\, \forall q<l, \,\, \forall j \in \mathcal{L_{\emph{q}}},  \,\, \pi_j(\hat{\theta}) = \pi_j^q(\hat{\theta}) 
		\end{align*}

        \STATE Set $\pi_i(\hat{\theta}) = \pi_i^l(\hat{\theta})$;

		\IF{$\pi_i^l(\hat{\theta}) \neq 0$}
		\STATE $\mathcal{K}^{remain} = \mathcal{K}^{remain} - \pi_i^l(\hat{\theta})$;
		\STATE $p_i(\hat{\theta}) = 	\mathcal{SW}_{-D_i}(\hat{\theta}) - (\mathcal{SW}_{-\mathcal{R_{\emph{l}}}}(\hat{\theta}) - \hat{v}_i(\pi_i^l(\hat{\theta})))$;
	
        \ELSE
		\STATE $p_i(\hat{\theta}) = 	\mathcal{SW}_{-D_i}(\hat{\theta}) - \mathcal{SW}_{-\mathcal{R_{\emph{l}}}}(\hat{\theta})$;
		\ENDIF
		\ENDFOR
		\IF{$\mathcal{K}^{remain} = 0$}
		\STATE Set $\pi_i(\hat{\theta}) = p_i(\hat{\theta}) = 0$, $\forall k>l, \forall i \in \mathcal{L_\emph{k}}(\hat{\theta})$
		\STATE break
		\ENDIF
		\ENDFOR
		\STATE Return $\pi_i(\hat{\theta})$ and $p_i(\hat{\theta})$ for each buyer $i$.
	\end{algorithmic}
	\label{NRM-T}
\end{breakablealgorithm}

Note that $\mathcal{SW}_{-\mathcal{R}_l}(\hat{\theta})$ and $\mathcal{SW}_{-D_i}(\hat{\theta})$ can be computed efficiently. As buyers have diminishing marginal valuations, consider the following greedy algorithm for selling $\mathcal{K}$ items. First sort all marginal valuations of all buyers, then allocate items to buyers corresponding to the largest $\mathcal{K}$ valuations. This algorithm is efficient and optimal.


In LDM-Tree, $\mathcal{R}_l(\hat{\theta})$ needs to be carefully designed to guarantee incentive compatibility. Intuitively, it should contain all \emph{potential competitors} of layer $l$ who are the buyers with positive \textit{potential utility} because these buyers have the motivation to act untruthfully to manipulate the allocation of high-priority buyers to earn. 

In our solution, we divide those potential competitors into the following two parts and design them for each buyer $i$: 
\begin{enumerate}
    \item Buyers who diffuse information to potential winners with large valuations, which is denoted as $C_i^{\mathcal{P}}(\hat{\theta})$ for buyer $i$. 
    \item Buyers who are potential winners, which is denoted as $C_i^{\mathcal{W}}(\hat{\theta})$ for buyer $i$.
\end{enumerate}

The corresponding sets are designed as follows:
\begin{enumerate}
    \item $C_i^{\mathcal{P}}(\hat{\theta})  = \{j| j \in C_i(\hat{\theta}), C_j(\hat{\theta}) \neq \emptyset\} \subseteq C_i(\hat{\theta})$ which is the children of $i$ who have children.
    \item $C_i^{\mathcal{W}}(\hat{\theta})$ is the top $\mathcal{K} + \mu - |C_i^{\mathcal{P}}(\hat{\theta})|$ ranked buyers in $C_i(\hat{\theta}) \setminus C_i^{\mathcal{P}}(\hat{\theta})$ according their valuation report for the first unit,  where $\mu$ is a constant with $\max_i |C_i^{\mathcal{P}}(\hat{\theta})|\le \mu$. 
    

\end{enumerate}

We assume $\mu$ is prior information for the seller which only depends on the structure of the network. It is the upper bound $|C_i^{\mathcal{P}}|$ for all $i$, and can be obtained from known network data. The reason why we introduce $\mu$ is to avoid buyers having incentives to switch from $C_i^{\mathcal{P}}(\hat{\theta})$ and  $C_i^{\mathcal{W}}(\hat{\theta})$. Such a phenomenon will happen if we only remove the top $\mathcal{K}$ ranked buyers for $C_i^{\mathcal{W}}(\hat{\theta})$ and will destroy incentive compatibility.

Let $C_i^{\mathcal{R}}(\hat{\theta}) = C_i^{\mathcal{W}}(\hat{\theta}) \cup C_i^{\mathcal{P}}(\hat{\theta})$ be the total removed set for buyer $i$. We consider layer $l$, we will remove all $C_i^{\mathcal{R}}(\hat{\theta})$ for $i \in  \mathcal{L}_l(\hat{\theta})$. Note that when this set is removed, all buyers in $\mathcal{L}_{\ge l+2}$ will not get information and also be removed because $C_i^{\mathcal{P}}(\hat{\theta})$ contains all children of $i$(in layer $l+1$) who have children. Thus, the final definition of $\mathcal{R_{\emph{l}}}(\hat{\theta})$ is  $\left(\bigcup_{i \in \mathcal{L_\emph{l}}(\hat{\theta})} C_i^{\mathcal{R}} (\hat{\theta}) \right) \cup \left( \bigcup_{\emph{l}+2 \le d \le l^{max}} \mathcal{L_{\emph{d}}}(\hat{\theta}) \right)$.

We give a notation table here for all complex notations given in the above discussion. These notations will be widely used in the following discussion.
\begin{center}
\begin{tabular}{ | m{3.5em} | m{18em}| } 
  \hline
  \textbf{Notation} & \textbf{Definition}\\ 
  \hline
  $C_i^{\mathcal{P}}(\hat{\theta})$ & The children of $i$ who have children. This set contains all buyers who can diffuse information to potential winners. \\
  \hline
  $\mu$ & A pre-known upper bound for $C_i^{\mathcal{P}}(\hat{\theta})$. $(\text{i.e.} \max_i |C_i^{\mathcal{P}}(\hat{\theta})| \le \mu))$ \\
  \hline
   $C_i^{\mathcal{W}}(\hat{\theta})$ &  The top $\mathcal{K} + \mu - |C_i^{\mathcal{P}}(\hat{\theta})|$ ranked buyers in $C_i(\hat{\theta}) \setminus C_i^{\mathcal{P}}(\hat{\theta})$ according their valuation report for the first unit. This set contains all potential winners.  \\ 
  \hline
  $C_i^{\mathcal{R}}(\hat{\theta})$ & $C_i^{\mathcal{R}}(\hat{\theta}) = C_i^{\mathcal{W}}(\hat{\theta}) \cup C_i^{\mathcal{P}}(\hat{\theta})$ which is the total removed set for buyer $i$. Removing this set can remove all potential competitors of buyer $i$.\\
  \hline
  $\mathcal{R_{\emph{l}}}(\hat{\theta})$ & $\left(\bigcup_{i \in \mathcal{L_\emph{l}}(\hat{\theta})} C_i^{\mathcal{R}} (\hat{\theta}) \right) \cup \left( \bigcup_{\emph{l}+2 \le d \le l^{max}} \mathcal{L_{\emph{d}}}(\hat{\theta}) \right)$ which is the total removed set when considering the allocation of layer $l$. This set contains all potential competitors of layer $l$.\\
  \hline
\end{tabular}
\end{center}

Here we give an example to illustrate our mechanism. We first show the removed sets through Figure \ref{figTree}. Suppose $\mathcal{K}=3$ and $\mu = 2$. Since $\hat{\theta}$ is given in the figure, we will omit it to simplify the description.
It is easy to see $C_f^{\mathcal{R}}=\{j\}, C_n^{\mathcal{R}}=\{q\}$ and $C_o^{\mathcal{R}}=\{r\}$. For the buyer $b$ with $C_b=\{d,e,f,g,h,i\}$, we have $C_b^{\mathcal{P}}=\{f,g\}$. Thus $\mu + \mathcal{K} - |C_b^{\mathcal{P}}| = 3$ i.e. $C_b^{\mathcal{W}} = \{d,e,h\}$. We have marked $C_b^{\mathcal{P}},C_b^{\mathcal{W}}$ in Figure \ref{figTree}. Taking the union of these two sets can get $C_{b}^{\mathcal{R}}=\{d,e,f,g,h\}$. For the buyer $g$ with $C_g=\{k,l,m,n,o,p\}$, we have $C_g^{\mathcal{P}}=\{n,o\}$, $C_g^{\mathcal{W}}=\{k,l,m\}$. $C_g^{\mathcal{P}}$ and $C_g^{\mathcal{W}}$ are also marked in Figure \ref{figTree}. Taking the union of these two sets can get $C_{g}^{\mathcal{R}}=\{k,l,m,n,o\}$. For other buyers, the removed sets are $\emptyset$.

\begin{figure}[htbp]
    \centering
    \includegraphics[scale=0.4]{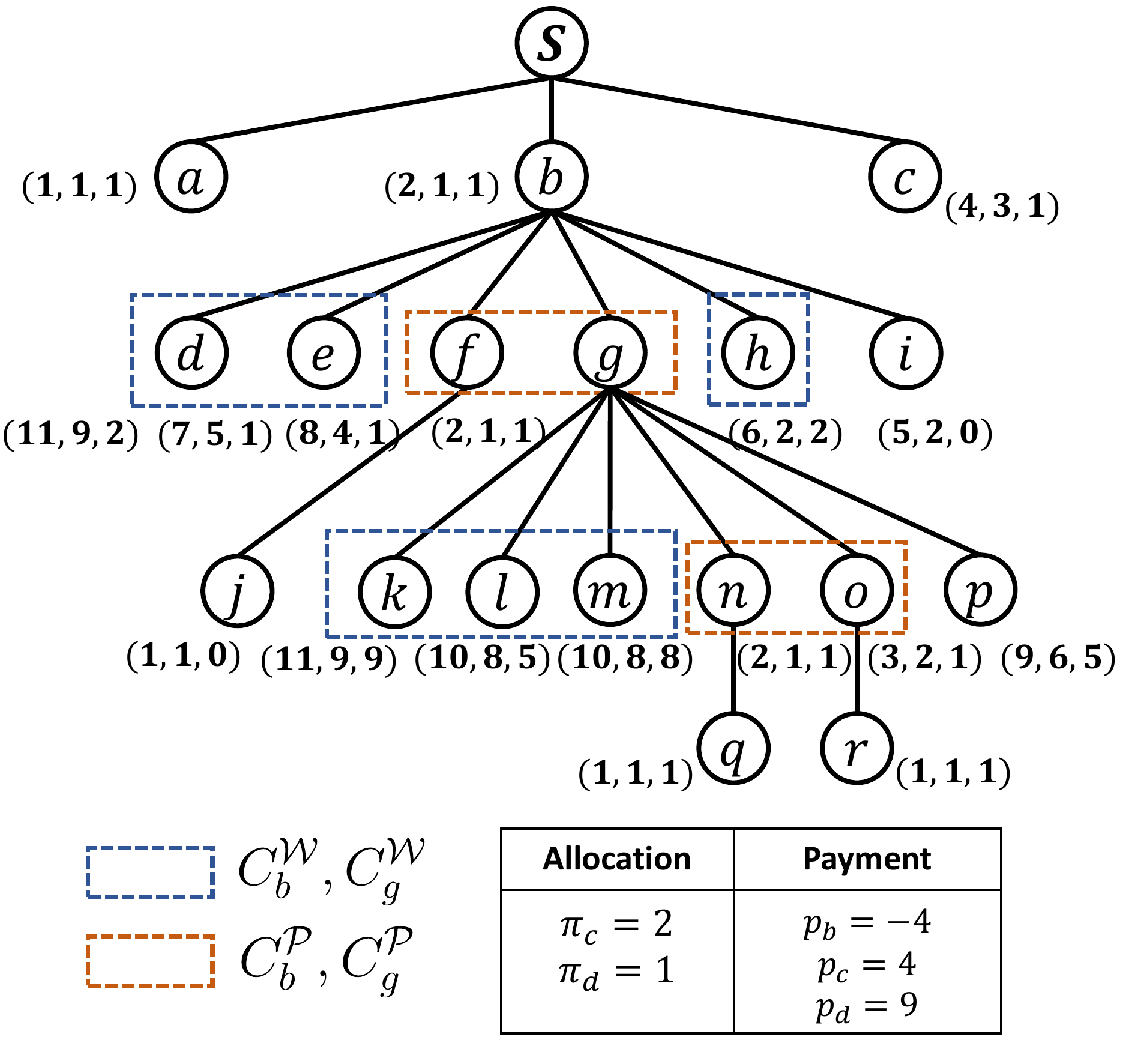}
    \caption{A tree example}
    \label{figTree}
\end{figure}

Then, let's compute the allocation and payment in Figure \ref{figTree}. For the first layer, the constrained optimization problem in step 4 of algorithm \ref{NRM-T} will be computed among $Q\setminus \mathcal{R}_{1}=\{a,b,c,i\}$. Since we now focus on the first layer, the constraint means that no items have been allocated. The solution is  $\pi^1_a=0,\pi^1_b=0,\pi^1_c=2, \pi^1_i=1$ and $\mathcal{SW}_{-\mathcal{R}_{\text{1}}}=5+4+3=12$. For buyer $a$, $\mathcal{SW}_{-D_a} = 5+4+3=12$ and her payment is $p_a=0$. For buyer $b$, $\mathcal{SW}_{-D_b} = 4+3+1=8$ and her payment is $p_b=8-12=-4$, i.e., buyer $b$ receives 4 from  the seller. For buyer $c$, $\mathcal{SW}_{-D_c} = 5+2+2=9$ and her payment is $p_c=9-(12-(4+3))=4$. Update $\mathcal{K}^{remain}=1$.

For the second layer, the constrained optimization problem in step 4 of algorithm \ref{NRM-T} will be computed among $Q \setminus \mathcal{R}_{\mathcal{L}_2} = \{a, b, c, d, e, f, g, h, i, p\}$. According to the allocation of the first layer, the constraint means that two items have been allocated to buyer $c$. There is only one item left and the solution is to allocate the last item to buyer $d$ , i.e., $\pi^2_d = 1$. Then $\mathcal{SW}_{-\mathcal{R}_{\text{2}}} = \hat{v}_c^1 + \hat{v}_c^2 + \hat{v}_d^1 = 4 + 3 + 11 = 18$. For buyer $d$, $\mathcal{SW}_{-D_d} = 9+4+3=16$ and her payment is $p_d=16-(18-11)=9$. For other buyers in $\mathcal{L}_2$, their payments are 0. Update $\mathcal{K}^{remain}=0$ and the auction is completed.

The seller's revenue is $9+4-4=9$ under LDM-Tree. If the seller only sells to her neighbors by VCG, it has $p_b^{VCG}=4+3+1-(4+3+2-2)=1$ and $p_c^{VCG}=2+1+1-(4+3+2-(4+3))=2$. The revenue of VCG is 3 and LDM-Tree achieves a higher revenue than VCG in this example. 

In the proofs of LDM-Tree's properties, we will formally analyze the necessity of the removed sets in detail. 


\begin{theorem} \label{t1}
The LDM-Tree is individually rational (IR).
\end{theorem}

\begin{proof}
If $\hat{\theta}_i = (v_i, \hat{r}_i)$ and $d_i = l$, then $u_i(\hat{\theta}_i, \hat{\theta}_{-i}) = \mathcal{SW}_{-\mathcal{R_{\emph{l}}}} - \mathcal{SW}_{-D_i}$ or $u_i(\hat{\theta}_i, \hat{\theta}_{-i}) = 0$. From the definition, the only difference between $\mathcal{SW}_{-D_i}$ and $\mathcal{SW}_{-\mathcal{R_{\emph{l}}}}$ is that different sets of buyers are considered when doing optimization. Since $Q \setminus D_i \subseteq Q \setminus \mathcal{R}_{l}, \mathcal{SW}_{-\mathcal{R_{\emph{l}}}} \ge  \mathcal{SW}_{-D_i}$. Thus $u_i((v_i, \hat{r}_i), \hat{\theta}_{-i}) \ge 0$.
\end{proof}

We can prove LDM-Tree is incentive compatible, which is given in Theorem \ref{t2}. Our proof is based on the key observations that buyers with positive potential utility have no incentive to change the removed set by misreporting valuations or stopping inviting. In that case, those buyers who are removed when computing lower layers' allocations have no incentives to affect the allocations to lower layers. The detailed analysis of such observation and the whole proof for IC are given in the appendix \ref{app:IC proof}.

\begin{theorem} \label{t2}
The LDM-Tree is incentive compatible (IC).
\end{theorem}


Although the exact social welfare of LDM-Tree depends on the tree structure, the following proposition gives the tight lower bound of the social welfare of LDM-Tree. 
\begin{proposition} \label{t3}
The social welfare of LDM-Tree is no less than the social welfare of the VCG in the first layer.
\end{proposition}

\begin{proof}
$\mathcal{SW}^{LDM} \ge \max\limits_{\pi} \sum\limits_{i \in Q \setminus \mathcal{R}_{1}} \hat{v}_i  \ge \max\limits_{\pi} \sum\limits_{i \in \mathcal{L}_{1}} \hat{v}_i  = \mathcal{SW}^{VCG}$
\end{proof}

The following proposition shows that LDM-Tree is non-wasteful whose proof is shown in the appendix \ref{app:nonwasterful-proof}.

\begin{proposition}
    LDM-Tree is non-wasteful.
\end{proposition}

In Theorem \ref{revenue}, we show the revenue improvements of LDM-Tree. The detailed proofs are postponed to the appendix \ref{app:revenue_proof}.
\begin{theorem}
The revenue of LDM-Tree is no less than the revenue of VCG mechanism in the first layer.
\label{revenue}
\end{theorem}


 As pointed in~\cite{kawasaki2020strategy}, the implementation of reserve price $R$ in diffusion settings is to add $\mathcal{K}$ dummy buyers with value $R$ in the first layer. Those dummy buyers never get items but are considered in every social welfare maximization problem to increase the price of true buyers. Actually, theorem 3 indicates that the LDM-Tree achieves no less revenue than VCG mechanism in the first layer even if adding more buyers in the first layer. Hence, setting a reserve price will not change the revenue quantity relationship.

This observation indicates the following corollary.

\begin{corollary}
The revenue of LDM-Tree with a reserve price $R$ is no less than the revenue of the VCG mechanism in the first layer with a reserve price $R$.
\end{corollary}

Note that there is no restriction on the reserve prices. In single-item auction, if buyers' valuation distribution is symmetric and regular, the optimal mechanism is a second price auction with a reserve price related to the distribution~\cite{myerson1981optimal}. Thus in the single-item Bayesian setting with symmetric and regular valuation distributions, LDM-Tree with reserve prices achieves better revenue than optimal auction without diffusion. For more general settings, VCG mechanism with reserve price also achieves a good approximation ratio to the optimal auction~\cite{hartline2009simple}, which also gives the approximation guarantee of LDM-Tree to the optimal mechanism without diffusion.

Here we clarify why we choose the VCG mechanism without diffusion as the benchmark rather than the optimal social welfare and revenue. It has been shown that in diffusion auction mechanism design, if IC, IR and weak budget balance are satisfied simultaneously, it is impossible to achieve efficiency at the same time~\cite{ZhangZC20}. Moreover, the loss of efficiency is unbounded. On the other hand, revenue maximization in diffusion settings has not been studied so far. Thus, following the previous work, we compare the social welfare and revenue with VCG mechanism without diffusion to attract the seller to apply the mechanism.

Although LDM-Tree relies on the prior knowledge of $\mu$, we point out that no properties will be affected if $\mu$ is overestimated. We can choose a rather large value if the prior is missing to guarantee diffusion IC and the allocation degenerates to VCG without diffusion. Moreover, $\mu$ is only related to the structure of the network, we have no restrictions on buyers' valuations. In the next section, we will extend the LDM-Tree to general graphs.

\section{Diffusion Auction Mechanism on Graphs}
In practice, most networks are general graphs. Similar to the tree network, given a report profile $\hat{\theta}$, a graph $\mathcal{G}(\hat{\theta}) = (V(\hat{\theta}) \cup \{s\}, E(\hat{\theta}))$ can be constructed where seller $s$ is the source of this graph and $V(\hat{\theta})$ is the set of valid buyers. For each $i \in V(\hat{\theta}) \cup \{s\}$ and her neighbor $j \in \hat{r}_i$, there is an edge $(i,j) \in E(\hat{\theta})$. 

LDM-Tree only works on tree-structure networks. To extend it to general graphs, we can transform a given graph $\mathcal{G}(\hat{\theta})$ into its corresponding breadth-first search tree(BFS tree) $\mathcal{T}^{BFS}(\hat{\theta})$ and run LDM-Tree on $\mathcal{T}^{BFS}(\hat{\theta})$. We will show such transformation can retain all properties of LDM-Tree. The transformation framework is given in Algorithm 3.

\begin{breakablealgorithm}
\caption{}
\begin{algorithmic}[1]
    \REQUIRE A report profile $\hat{\theta}$ and a diffusion mechanism $\mathcal{M}$ for trees;
	\ENSURE  $\pi(\hat{\theta})$ and $p(\hat{\theta})$;
	
	\STATE Construct the graph $\mathcal{G}(\hat{\theta})$;
	\STATE Run breadth-first search on $\mathcal{G}(\hat{\theta})$ in alphabetical order to form  $\mathcal{T}^{BFS}(\hat{\theta})$ ;
	\STATE Run $\mathcal{M}$ on $\mathcal{T}^{BFS}(\hat{\theta})$;
	\STATE Return $\pi_i(\hat{\theta})$ and $p_i(\hat{\theta})$ for each buyer $i$.
\end{algorithmic}
\label{GTT}
\end{breakablealgorithm}

If the input mechanism $\mathcal{M}$ for Algorithm~\ref{GTT} is the LDM-Tree, then we call the whole algorithm LDM. Now, let's run the LDM on an example shown in Figure \ref{figGraph}. Figure \ref{figGraph} is a graph network. LDM first transforms it into its BFS tree. The BFS tree is exactly the same as the tree in Figure \ref{figTree}. LDM then runs LDM-Tree on the BFS tree to find the allocation and payment for each buyer which are shown in Figure \ref{figTree}. 

\begin{figure}[htbp]
    \centering
    \includegraphics[scale=0.4]{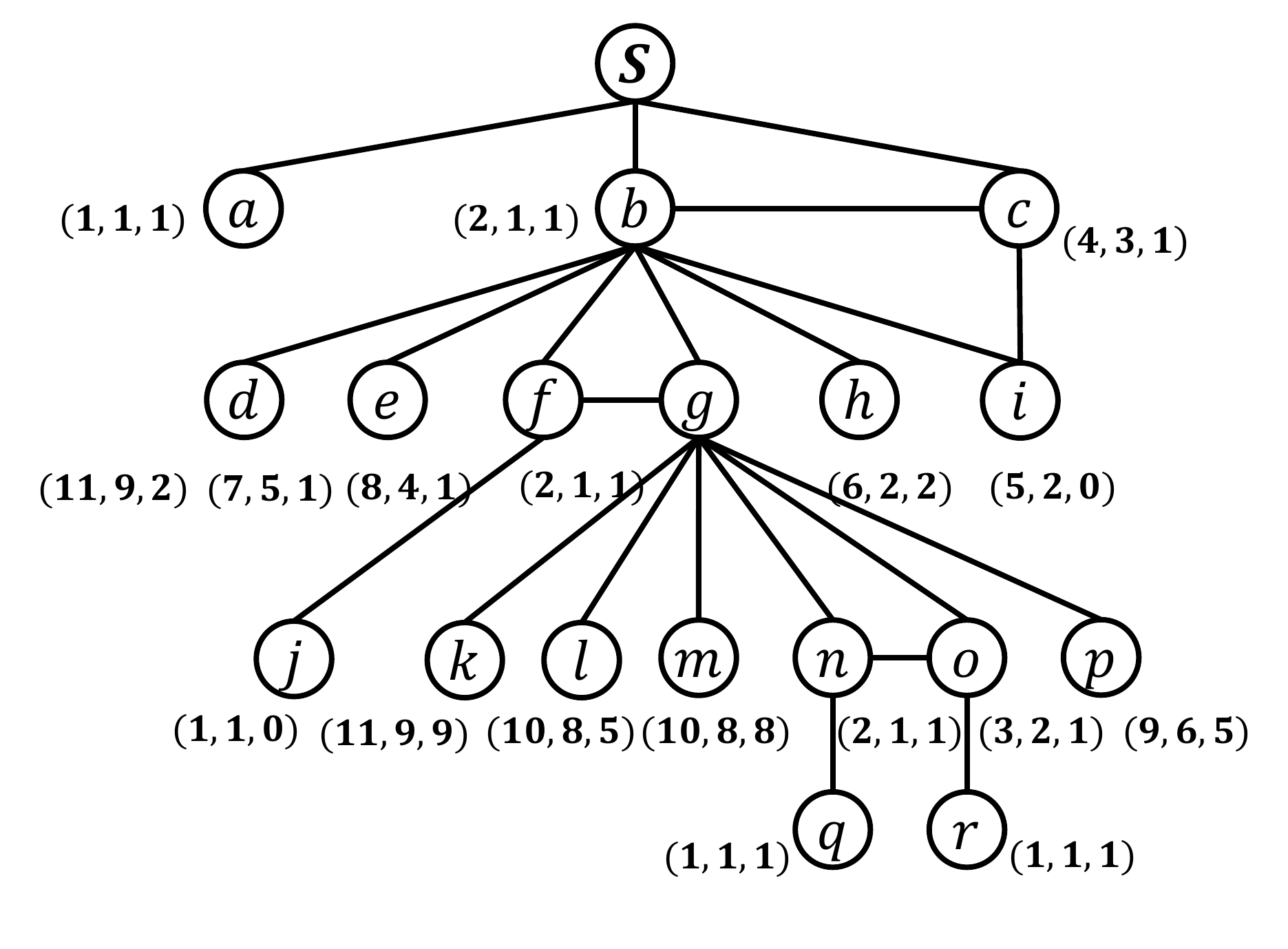}
    \caption{A graph example}
    \label{figGraph}
\end{figure}



In order to ensure that Algorithm \ref{GTT} is incentive compatible, the input mechanism $\mathcal{M}$ should be incentive compatible first to ensure truthful valuation report of each buyer. Besides, it should satisfy that each buyer's utility is non-increasing when another buyer in the same layer has more children to guarantee invitation incentive. Let $u_i(C_j)$ be the utility of $i$ when the children of $j$ is $C_j$. The following theorem summarizes the above discussion, and its proof is given in the appendix \ref{G-to-T proof}.

\begin{theorem} \label{t5}
If the input mechanism $\mathcal{M}$ for Algorithm~\ref{GTT} is IC and for any two buyers $i,j \in \mathcal{L}_l$, $u_i(\hat{C}_j) \ge u_i(C_j)$ where $\hat{C}_j \subseteq C_j$, then Algorithm~\ref{GTT} is IC.
\end{theorem}

Based on Theorem \ref{t5}, we can finally know LDM has all desired properties. The detailed proof of Corollary \ref{LDM} is postponed to the appendix \ref{G-to-T proof}.
\begin{corollary}
The LDM on graphs has the following properties. (1) Incentive Compatible. (2) Individually Rational. (3) Non-wasteful. (4) The social welfare/revenue of LDM is no less than the social welfare/revenue of the VCG mechanism in the first layer
\label{LDM}
\end{corollary}






From Theorem \ref{t5}, if we want to design an IC diffusion mechanism on graphs, it is enough to design an IC diffusion mechanism on trees with one additional property. Moreover, the properties related to individual rationality, social welfare and revenue of Algorithm \ref{GTT} are the same as its input mechanism $\mathcal{M}$. Since the tree structure is much simpler than graphs and the additional property is natural, Theorem \ref{t5} can reduce the difficulty of designing diffusion mechanisms on graphs, which may allow us to handle more intricate settings in diffusion mechanism design.

\section{Conclusions}
We designed the layer-based diffusion mechanism (LDM) for multi-unit auctions with diminishing marginal utility buyers. In the LDM, buyers are classified into different layers according to their shortest distance to the seller. The LDM will compute allocation and payment layer by layer based on a constrained optimization problem of social welfare. The LDM can incentivize buyers to invite their neighbors to join the auction and truthfully report their valuations. We also showed that the seller can get higher revenue by using LDM than only using VCG (with reserve price) in the first layer. It is the very first diffusion mechanism to satisfy such desirable properties for multi-unit-demand settings. 

In addition, we show the effectiveness of BFS tree in diffusion mechanism design. With BFS tree, designing IC mechanisms on trees with one extra natural property is sufficient to design IC mechanisms on graphs. We argue that this transformation can simplify the design of diffusion mechanism on graphs.

There are many interesting future directions to do based on LDM. We proposed a general framework to design diffusion auctions for buyers with diminishing marginal utilities. The only thing we need to optimize is $C_i^{\mathcal{R}}$. It should contain all buyers with positive potential utility and buyers in $C_i^{\mathcal{R}}$ should have no incentive to change this set. All $C_i^{\mathcal{R}}$ satisfies above two properties will lead to an IC, IR mechanism with Algorithm $1$. One trivial method is to set $C_i^{\mathcal{R}} = C_i$, which is equivalent to VCG mechanism in the first layer. We proposed a nontrivial method to design $C_i^{\mathcal{R}}$ to achieve better social welfare and revenue than the trivial case, but it relies on one additional prior knowledge $\mu$. Designing diffusion mechanisms without $\mu$ but still have social welfare and revenue improvements is an interesting future work, which requires a refined understanding of potential competitors. Moreover, LDM fails to be extended to more general combinatorial settings. Finding new methods to design diffusion combinatorial auctions will be the ultimate goal of this research line.



\bibliographystyle{named} 
\bibliography{ijcai23}

\newpage
\noindent \textbf{\huge{Appendix}}
~\\
\begin{appendix}
\noindent In the following analysis, we will use the notations given in the notation table.
\begin{center}
\begin{tabular}{ | m{3.5em} | m{18em}| } 
  \hline
  \textbf{Notation} & \textbf{Definition}\\ 
  \hline
  $C_i^{\mathcal{P}}(\hat{\theta})$ & The children of $i$ who have children.\\
  \hline
 $\mu$ & A pre-known upper bound for $C_i^{\mathcal{P}}(\hat{\theta})$. $(\text{i.e.} \max_i |C_i^{\mathcal{P}}(\hat{\theta})|\le \mu)$\\ 
  \hline
   $C_i^{\mathcal{W}}(\hat{\theta})$ &  The top $\mathcal{K} + \mu - |C_i^{\mathcal{P}}(\hat{\theta})|$ ranked buyers in $C_i(\hat{\theta}) \setminus C_i^{\mathcal{P}}(\hat{\theta})$ according their valuation report for the first unit.  \\ 
  \hline
  $C_i^{\mathcal{R}}(\hat{\theta})$ & $C_i^{\mathcal{R}}(\hat{\theta}) = C_i^{\mathcal{W}}(\hat{\theta}) \cup C_i^{\mathcal{P}}(\hat{\theta})$ which is the total removed set for buyer $i$.\\
  \hline
  $\mathcal{R_{\emph{l}}}(\hat{\theta})$ & $\left(\bigcup_{i \in \mathcal{L_\emph{l}}(\hat{\theta})} C_i^{\mathcal{R}} (\hat{\theta}) \right) \cup \left( \bigcup_{\emph{l}+2 \le d \le l^{max}} \mathcal{L_{\emph{d}}}(\hat{\theta}) \right)$ which is the total removed set when considering the allocation of layer $l$\\
  \hline
    $D_i(\hat{\theta})$ & $D_i(\hat{\theta}) = \mathcal{R}_l(\hat{\theta}) \cup C_i(\hat{\theta}) \cup \{i\}$ which is removed to exclude all influence of buyer $i$.\\
  \hline
\end{tabular}
\end{center}

\section{Missing proof in Section 4}
\subsection{Proofs Related to Incentive Compatibility}\label{app:IC proof}
We first show the following property that buyers with positive potential utility have no incentive to change the removed set by misreporting valuations.
\begin{observation}
For any buyer $j$ and $i \in r_j$, given any $\hat{\theta}_{-i}$, if $i \in C_j^{\mathcal{R}}((v_i , \hat{r}_i), \hat{\theta}_{-i})$ and $i \in C_j^{\mathcal{R}}((\hat{v}_i , \hat{r}_i), \hat{\theta}_{-i})$, then $C_j^{\mathcal{R}}((\hat{v}_i , \hat{r}_i), \hat{\theta}_{-i})= C_j^{\mathcal{R}}((v_i , \hat{r}_i), \hat{\theta}_{-i})$.
\label{O1}
\end{observation}

\begin{proof}
\textbf{Case 1.} If $i \in C_j^{\mathcal{P}}$ when $\hat{\theta}_i = (v_i, \hat{r}_i)$, $i$ is still in $C_j^{\mathcal{P}}$ after changing report type to $(\hat{v}_i, \hat{r}_i)$. Then all sets remain the same and $C_j^{\mathcal{R}}$ will not change.

\textbf{Case 2.} If $i \in C_j^{\mathcal{W}}$ when $\hat{\theta}_i = (v_i, \hat{r}_i)$. Then after changing type to $(\hat{v}_i, \hat{r}_i)$, $i$ can not be in $C_j^{\mathcal{P}}$. Given $i \in C_j^{\mathcal{R}}(\hat{v}_i , \hat{r}_i)$, $i \in C_j^{\mathcal{W}}(\hat{v}_i , \hat{r}_i)$. Then all sets remain the same and $C_j^{\mathcal{R}}$ will not change.
\end{proof}

Using Observation \ref{O1}, in Lemma \ref{l3}, we prove that in LDM-Tree, reporting truthfully is always a dominant strategy for every agent no matter how many neighbors she diffuses information to.
\begin{lemma} \label{l3}
Given any $\hat{\theta}_{-i}$ and fix the invitation of $i$ to be $\hat{r}_i$, we have $u_i((v_i, \hat{r}_i), \hat{\theta}_{-i}) \ge u_i((\hat{v}_i, \hat{r}_i), \hat{\theta}_{-i})$ for any $\hat{v}_i$.
\end{lemma}

\begin{proof}
Suppose $i \in \mathcal{L_{\emph{l}}}$ and her parent is $j$.

\textbf{Case 1.}  $i \in C_j^{ \mathcal{R}}$ when $i$ reports truthfully.

$i$ has no incentive to misreport to be in $C_j \setminus C_j^{\mathcal{R}}$ because in that case $u_i = 0$ after misreporting. If $i$ misreports to $\hat{v}_i$ and still in $C_j^{\mathcal{R}}$, from Observation 1, buyers in $ C_j \setminus C_j^{\mathcal{R}}$ will not change which indicates that for $l > 1$ and $1 \le d \le l-1$, allocation $\pi^{d}$  will not change. Thus $\mathcal{SW}_{-\mathcal{R_{\emph{l}}}}(v_i, \hat{r}_i)$ and $\mathcal{SW}_{-\mathcal{R_{\emph{l}}}}(\hat{v}_i, \hat{r}_i)$ are computed under the same constraints and so do $\mathcal{SW}_{-D_i}(v_i, \hat{r}_i)$ and $\mathcal{SW}_{-D_i}(\hat{v}_i, \hat{r}_i)$. When computing $\mathcal{SW}_{-D_i}$, buyer $i$ is removed and we have shown the constraints are independent of $\hat{v}_i$. Thus $\mathcal{SW}_{-D_i}$ is a constant given $\hat{\theta}_{-i}$. If $i$ reports truthfully, no matter what the $\pi_i^l(v_i, \hat{r}_i)$ is, $u_i(v_i, \hat{r_i}) = \mathcal{SW}_{-\mathcal{R_{\emph{l}}}}(v_i, \hat{r}_i) - \mathcal{SW}_{-D_i}$. 

If $i$ misreports valuation and $\pi_i^l(\hat{v}_i, \hat{r}_i) = \pi_i^l(v_i, \hat{r}_i)$, then $u_i(v_i, \hat{r}_i) = u_i(\hat{v}_i, \hat{r}_i)$. If $i$ misreports and $\pi_i^l(\hat{v}_i, \hat{r}_i) \neq \pi_i^l(v_i, \hat{r}_i)$,  then $u_i(\hat{v}_i, \hat{r_i}) = v_i(\pi_i^l(\hat{v}_i, \hat{r}_i)) - (\mathcal{SW}_{-D_i} - (\mathcal{SW}_{-\mathcal{R_{\emph{l}}}}(\hat{v}_i, \hat{r}_i) - \hat{v}_i(\pi_i^l(\hat{v}_i, \hat{r}_i))) =  v_i(\pi_i^l(\hat{v}_i, \hat{r}_i)) + \mathcal{SW}_{-\mathcal{R_{\emph{l}}}}(\hat{v}_i, \hat{r}_i) - \hat{v}_i(\pi_i^l(\hat{v}_i, \hat{r}_i)) - \mathcal{SW}_{-D_i}$.

If $v_i(\pi_i^l(\hat{v}_i, \hat{r}_i)) + \mathcal{SW}_{-\mathcal{R_{\emph{l}}}}(\hat{v}_i, \hat{r}_i) - \hat{v}_i(\pi_i^l(\hat{v}_i, \hat{r}_i)) \ge \mathcal{SW}_{-\mathcal{R_{\emph{l}}}}(v_i, \hat{r}_i)$, then it contradicts to that $\mathcal{SW}_{-\mathcal{R_{\emph{l}}}}(v_i, \hat{r}_i)$ is optimal. Thus $u_i(v_i, \hat{r_i}) \ge u_i(\hat{v}_i, \hat{r_i})$.

\textbf{Case 2.}  $i \notin C_j^{\mathcal{R}}$ when $i$ reports truthfully.

Since $i \notin  C_j^{\mathcal{R}}$ when reporting truthfully, $\pi_i^l(v_i, \hat{r}_i) = 0$ and $u_i(v_i, \hat{r}_i) = 0$.  If $i$ misreports to $\hat{v}_i$ and does not get any items, then $u_i(\hat{v}_i, \hat{r}_i) = 0$. If $i$ misreports to $\hat{v}_i$ and $\pi_i^l(\hat{v}_i, \hat{r}_i) \neq 0$, then $u_i(\hat{v}_i, \hat{r}_i) =  v_i(\pi^l_i(\hat{v}_i, \hat{r}_i)) + \mathcal{SW}_{-\mathcal{R_{\emph{l}}}}(\hat{v}_i, \hat{r}_i) - \hat{v}_i(\pi^l_i(\hat{v}_i, \hat{r}_i)) - \mathcal{SW}_{-D_i} \le 0 = u_i(v_i, \hat{r}_i)$ because there are at least $\pi_i^l(\hat{v}_i, \hat{r}_i)$ buyers in $\mathcal{L_{\emph{l}}}$ whose value of the first item is larger than $i$'s.
\end{proof}

The following property shows that buyers with positive potential utility have no incentive to change the removed set by stopping information diffusion.

\begin{observation}
For any buyer $j$ and $i \in r_j$, given any $\hat{\theta}_{-i}$, if $i \in C_j^{\mathcal{R}}((\hat{v}_i , r_i), \hat{\theta}_{-i})$ and $i \in C_j^{\mathcal{R}}((\hat{v}_i , \hat{r}_i), \hat{\theta}_{-i})$, then $C_j^{\mathcal{R}}((\hat{v}_i , \hat{r}_i), \hat{\theta}_{-i}) = C_j^{\mathcal{R}}((\hat{v}_i , r_i), \hat{\theta}_{-i})$.
\label{O2}
\end{observation}

\begin{proof}
\textbf{Case 1:} If $i \in C_j^{\mathcal{P}}$ when $\hat{\theta}_i = (\hat{v_i}, r_i)$, $i$ may belong to two different sets after changing report type to $(\hat{v}_i, \hat{r}_i)$.

(1) If $i \in C_j^{\mathcal{P}}$, then no sets will change and $C_j^{\mathcal{R}}$ will not change as well.

(2) If $i \in C_j^{\mathcal{W}}$, then $|C_j^{\mathcal{P}}|$ will decrease by $1$ and $\mu + \mathcal{K} - |C_j^{\mathcal{P}}|$ will increase by $1$. The size increase of $C_j^{\mathcal{W}}$ will absorb $i$, which makes $C_j^{\mathcal{R}}$ unchanged.

\textbf{Case 2:} If $i \in C_j^{\mathcal{W}}$ when $\hat{\theta}_i = (\hat{v_i}, r_i)$, then $i$ do not have children, which means $\hat{r}_i = r_i = \emptyset$. In that case, $C_j^{\mathcal{R}}$ will not change. 
\end{proof}

Using Observation \ref{O2}, in Lemma \ref{l4}, we prove that in LDM-Tree, for any buyer $i$, given $i$ reports valuation truthfully, diffusing information to all neighbors is a dominant strategy.
\begin{lemma}  \label{l4}
Given any $\hat{\theta}_{-i}$ and fix the valuation report of $i$ to be $v_i$, for any $\hat{r}_i \subseteq r_i$, we have $u_i((v_i, r_i),  \hat{\theta}_{-i}) \ge u_i((v_i, \hat{r}_i), \hat{\theta}_{-i})$.
\end{lemma}

\begin{proof}
Suppose $i \in \mathcal{L_{\emph{l}}}$ and her parent is $j$. If $r_i \neq \emptyset$, then when $i$ diffuses information to all neighbors, $i \in C_j^{\mathcal{R}}$. If $\hat{r}_i \neq \emptyset$, $i$ is still in $C_j^{\mathcal{R}}$ after diffusing the information to less neighbors. By Observation 2, no matter how many neighbors $i$ diffuses information to,  $\mathcal{SW}_{-\mathcal{R_{\emph{l}}}}(v_i, \hat{r}_i)$ and $\mathcal{SW}_{-D_i}$ are computed under the same constraints. When computing $\mathcal{SW}_{-D_i}$, buyers in $C_i$ are all removed. Thus $\mathcal{SW}_{-D_i}$ is a constant given $\hat{\theta}_{-i}$. $u_i(v_i, \hat{r}_i) = \mathcal{SW}_{-\mathcal{R_{\emph{l}}}}(v_i, \hat{r}_i) - \mathcal{SW}_{-D_i}$ where the first item of is monotonic increasing with $\hat{r}_i$ while the second item is independent of $\hat{r}_i$. If $\hat{r_i} = \emptyset$, then $i$ can not control any diffusion. So diffusing information to all neighbors will maximize the utility of $i$. 
\end{proof}

\setcounter{theorem}{1}
Combining Lemma 1 and Lemma 2, we prove LDM-Tree is IC.
\begin{theorem}
    The LDM-Tree is incentive compatible (IC).
\end{theorem}

\begin{proof}
Given $\hat{\theta}_{-i}$, for all $(\hat{v}_i, \hat{r}_i)$, we have $u_i(v_i, r_i) \ge u_i(v_i, \hat{r}_i) \ge u_i(\hat{v}_i, \hat{r}_i)$ where the first inequality comes from Lemma ~\ref{l4} and the second inequality comes from Lemma ~\ref{l3}. So LDM-Tree is IC. 
\end{proof}

\setcounter{proposition}{1}
\subsection{Proofs Related to nonwastfulness}\label{app:nonwasterful-proof}
\begin{proposition}
    LDM-Tree is non-wasteful.
\end{proposition}

\begin{proof}
Let $l^{max}$ be the number of layers. If there are $m$ remaining items when we compute the allocation for the last layer, then there is a set $M \subset \mathcal{L}_l(\hat{\theta})$ such that for all $i \in M$, $\pi^{l^{max}-1}_i(\hat{\theta}) = 1$.  Let $W^{l^{max} - 1} = \bigcup _i C_i^{\mathcal{W}}(\hat{\theta})$ for all $i \in \mathcal{L}_{l^{max} - 1}(\hat{\theta})$. Since for all $i \in M$, there is a buyer $j \in W^{l^{max} - 1}$ such that $v_j^1 \ge v_i^1$, then all $m$ remaining items will be allocated to buyers in $W^{l^{max} - 1}  \subset \mathcal{L}_l(\hat{\theta})$. In that case, all items will be allocated in the end. If no items remain for the last layer, then all items are allocated to lower layers. Thus, LDM-Tree is non-wasteful. 
\end{proof}

\subsection{Proofs Related to Revenue Improvement} \label{app:revenue_proof}
Now we will show the revenue improvement of LDM-Tree. Let $\pi_i^{VCG}, p_i^{VCG}$ be the allocation and payment of $i$ under VCG mechanism among the first layer. Let $p_i^{LDM}$ be the payment of $i$ under LDM-Tree mechanism. Let $\phi(z) = \mathcal{K} - z + 1$,  $ \sum_{k = \phi(z)}^{\mathcal{K}} v_{S}^k$ is the sum of the last $z$ values in top $\mathcal{K}$ ranked value in set $S$. In the following analysis, we use $-D_i$ as the abbreviation for $Q \setminus D_i$. For $i \in \mathcal{L_{\emph{l}}}$, let $m_i^{LDM} = \pi_i^l + \sum_{j \in C_i} \pi_j^l$. Let $q_i^{LDM} =  \mathcal{SW}_{-D_i} - \sum_{j \notin C_i \cup \{i\}} v_j(\pi_j^l) =  \sum_{k = \phi(m_i^{LDM})}^{\mathcal{K}} v_{-D_i}^k$. Let $t_i^{LDM} = \sum_{j \in C_i} v_j(\pi_j^l)$. Then $p_i^{LDM} = q_i^{LDM} - t_i^{LDM}.$

We will first show the sum of $q_i^{LDM}$ for all buyers in $\mathcal{L_{\text{1}}}$ is larger than revenue of VCG mechanism among $\mathcal{L_{\text{1}}}$. Then, we will show the sum of $t_i^{LDM}$ for $i \in \mathcal{L_{\emph{l} - \text{1}}}$ can be compensated by the sum of $q_i^{LDM}$ for $i \in \mathcal{L_{\emph{l}}}$ for any $l \ge 2$. Combining these two results, we will show the revenue of LDM-Tree is better than the revenue of VCG mechanism among the seller's neighbors.

\begin{observation}
$\sum_{i \in \mathcal{L_{\emph{1}}}} q_i^{LDM} \ge \sum_{i \in \mathcal{L_{\emph{1}}}} p_i^{VCG}$.
\end{observation}

\begin{proof}
We divide buyers in $\mathcal{L_{\text{1}}}$ into three groups $\alpha_1, \beta_1$ and $\gamma_1$ according to $\pi_i^{VCG}$ and  $m_i^{LDM}$.

Let all buyers $i$ with $\pi_i^{VCG} < m_i^{LDM} $ be in the group $\alpha_1$. For these buyers, we have 
\begin{align*}
p_i^{VCG} &= \sum_{k = \phi(\pi_i^{VCG})}^{\mathcal{K}}v^k_{\mathcal{L_{\text{1}}} \setminus \{i\}} 
    \\&\le \sum_{k =  \phi( \pi_i^{VCG}) } ^{\mathcal{K}} v_{-D_i}^k  
    \\&\le \sum_{k =  \phi( \pi_i^{VCG}) } ^{\mathcal{K}} v_{-D_i}^k  + \sum_{k =  \phi(m_i^{LDM}) } ^{\mathcal{K} - \pi_i^{VCG}} v_{-D_i}^k = q_i^{LDM}
\end{align*}

Let all buyers $i$ with $\pi_i^{VCG} = m_i^{LDM} $ be in the group $\beta_1$. For these buyers, we have 
\begin{align*}
p_i^{VCG} &= \sum_{k = \phi(\pi_i^{VCG})}^{\mathcal{K}}v^k_{\mathcal{L_{\text{1}}} \setminus \{i\}} 
\\&\le \sum_{k =  \phi( \pi_i^{VCG}) } ^{\mathcal{K}} v_{-D_i}^k  
\\&= \sum_{k =  \phi( m_i^{LDM}) } ^{\mathcal{K}} v_{-D_i}^k = q_i^{LDM}    
\end{align*}

Let all buyers $i$ with $\pi_i^{VCG} > m_i^{LDM} $ be in the group $\gamma_1$ and  $\Delta_i = \pi_i^{VCG} - m_i^{LDM}$. Then 
\begin{equation*}
p_i^{VCG} = \sum_{k = \phi(\pi_i^{VCG})}^{\mathcal{K}}v^k_{\mathcal{L_{\text{1}}} \setminus \{i\}} \le \sum_{k = \phi(m_i^{LDM})}^{\mathcal{K}} v_{-D_i}^k + \sum_{k=\pi_i^1 + 1}^{\pi_i^1 + \Delta_i} v_i^k
\end{equation*}
Since $\sum\limits_{i \in \alpha_1} \sum_{k =  \phi(m_i^{LDM}) } ^{\mathcal{K} - \pi_i^{VCG}} v_{-D_i}^k \ge \sum\limits_{i \in \gamma_1}\sum_{k=\pi_i^1 + 1}^{\pi_i^1 + \Delta_i} v_i^k$, we have
\begin{align*}
    &\sum\limits_{i \in \alpha_1} p_i^{VCG} + \sum\limits_{i \in \gamma_1} p_i^{VCG}
    \\&\le \sum\limits_{i \in \alpha_1}\sum_{k =  \phi( \pi_i^{VCG}) } ^{\mathcal{K}} v_{-D_i}^k 
    \\&+\sum\limits_{i \in \gamma_1}(\sum_{k=\mathcal{K}-m_i^{LDM} + 1}^{\mathcal{K}} v_{-D_i}^k + \sum_{k=\pi_i^1 + 1}^{\pi_i^1 + \Delta_i} v_i^k) 
    \\&\le \sum\limits_{i \in \alpha_1}(\sum_{k =  \phi( \pi_i^{VCG}) } ^{\mathcal{K}} v_{-D_i}^k  + \sum_{k =  \phi(m_i^{LDM}) } ^{\mathcal{K} - \pi_i^{VCG}} v_{-D_i}^k) 
    \\&+  \sum\limits_{i \in \gamma_1} \sum_{ k= \phi(m_i^{LDM}) }^{\mathcal{K}}v_{\mathcal{L_{\text{1}}} \setminus \{i\}}^k
    \\ &\le \sum\limits_{i \in \alpha_1}q_i^{LDM} + \sum\limits_{i \in \gamma_1} q_i^{LDM}
\end{align*}
Combining with buyers in $\beta_1$, we get $\sum_{i \in \mathcal{L_{\text{1}}}} p_i^{VCG} \le  \sum_{i \in \mathcal{L_{\text{1}}}} q_i^{LDM}$. 

\end{proof}

\begin{observation}
For any layer $l \ge 2$, $\sum\limits_{i \in \mathcal{L}_{\text{l}}} q_i^{LDM} \ge \sum\limits_{i \in \mathcal{L_{\text{l} - \emph{1}}}} t_i^{LDM}$.
\end{observation}

\begin{proof}
For any $i \in \mathcal{L_{\emph{l} - \text{1}}}$ and $j \in C_i$ with $\pi_j^{l-1} \neq 0$, $j$ can't be the top $\mathcal{K}$ buyers in $C_i$ according to the report value for the first item. When computing $q_i^{LDM}$, there are at least $\mathcal{K}$ values that are larger than $v_j^1$. Thus $\sum\limits_{i \in \mathcal{L}_{\emph{l}}} q_i^{LDM} \ge \sum\limits_{i \in \mathcal{L_{\emph{l} - \text{1}}}} t_i^{LDM}$. 
\end{proof}

\begin{theorem}
The revenue of LDM-Tree is no less than the revenue of VCG mechanism in the first layer.
\end{theorem}

\begin{proof}
Let $\mathcal{L}_{\emph{f}}$ be the last layer in which there are buyers who get items. Since no items are allocated in $\mathcal{L_{\emph{f} + \text{1}}}$, $\sum_{i \in \mathcal{L_{\emph{f}}}} t_i^{LDM} = 0$. Then $\sum_{i \in Q} p_i^{LDM} = \sum_{i \in \mathcal{L_{\text{1}}}} p_i^{LDM} + \sum_{l=2}^{f} \sum_{i \in \mathcal{L_{\emph{l}}}} p_i^{LDM} =  \sum_{i \in \mathcal{L_{\text{1}}}} q_i^{LDM} + \sum_{l=2}^{f} (\sum_{i \in \mathcal{L_{\emph{l}}}} q_i^{LDM} - \sum_{i \in \mathcal{L_{\emph{l}- \text{1}}}} t_i^{LDM})$. By Observation 3, $\sum_{i \in \mathcal{L_{\text{1}}}} q_i^{LDM}  \ge \sum_{i \in \mathcal{L_{\text{1}}}} p_i^{VCG}$. By Observation 4, $\sum_{l=2}^{f} (\sum_{i \in \mathcal{L_{\emph{l}}}} q_i^{LDM} - \sum_{i \in \mathcal{L_{\emph{l}- \text{1}}}} t_i^{LDM})  \ge 0$. Combine all above inequalities,  $\sum_{i \in Q} p_i^{LDM} \ge  \sum_{i \in \mathcal{L_{\text{1}}}} p_i^{VCG}$.
\end{proof}

\section{Missing Proofs in Section 5} \label{G-to-T proof}
\begin{breakablealgorithm}
\caption{}
\begin{algorithmic}[1]
    \REQUIRE A report profile $\hat{\theta}$ and a diffusion mechanism $\mathcal{M}$ for trees;
	\ENSURE  $\pi(\hat{\theta})$ and $p(\hat{\theta})$;
	
	\STATE Construct the graph $\mathcal{G}(\hat{\theta})$;
	\STATE Run breadth-first search on $\mathcal{G}(\hat{\theta})$ in alphabetical order to form  $\mathcal{T}^{BFS}(\hat{\theta})$ ;
	\STATE Run $\mathcal{M}$ on $\mathcal{T}^{BFS}(\hat{\theta})$;
	\STATE Return $\pi_i(\hat{\theta})$ and $p_i(\hat{\theta})$ for each buyer $i$.
\end{algorithmic}
\label{GTT}
\end{breakablealgorithm}

\begin{theorem} \label{t5}
If the input mechanism $\mathcal{M}$ for Algorithm~\ref{GTT} is IC and for any two buyers $i,j \in \mathcal{L}_l$, $u_i(\hat{C}_j) \ge u_i(C_j)$ where $\hat{C}_j \subseteq C_j$, then Algorithm~\ref{GTT} is IC.
\end{theorem}

\begin{proof}
Since $\mathcal{T}^{BFS}(\hat{\theta})$ in Algorithm~\ref{GTT} only depends on buyers' invitations, buyers will report valuation truthfully following the IC of $\mathcal{M}$. Buyers in the same layer can not influence each other by invitation because they have the same shortest distance to the seller. If a buyer $i$ does not diffuse information to her neighbor $w$, then $w$ can not be $i$'s child in $\mathcal{T}^{BFS}(\hat{\theta})$. Instead, $w$ will become the child of $j$ who is in the same layer as $i$ with $w \in r_j$ . Since $\mathcal{M}$ is IC, diffusing information to all neighbors is a dominant strategy which indicates that $u_i$ will be worse if $i$ has fewer children. In addition, $u_i$ is non-increasing when $j$ has more children. Thus $i$ has incentives to diffuse information and Algorithm~\ref{GTT} is IC. 
\end{proof}

\begin{corollary}
The LDM on graphs has the following properties. (1) Incentive Compatible. (2) Individually Rational. (3) Non-wasteful. (4) The social welfare of LDM is no less than the social welfare of the VCG mechanism in the first layer. (5) The revenue of LDM is no less than the revenue of the VCG mechanism in the first layer.
\end{corollary}

\begin{proof}
By Theorem 1, Proposition 1, Proposition 2, and Theorem 3, LDM-Tree is IR, non-wasteful, and achieves social welfare and revenue no less than the VCG mechanism in the first layer. The LDM has the same properties because it finally runs LDM-Tree on $\mathcal{T}^{BFS}(\hat{\theta})$. To prove LDM is IC, we only need to prove LDM-Tree satisfies Theorem \ref{t5}.

Assume there are two buyers $i,j \in \mathcal{L_{\emph{l}}}$. Let $C_j^{out} = C_j \setminus C_j^{\mathcal{R}}$. When $j$ has one more child $w$, $C_j^{out}$ will have one more buyer or remain the same.

Before $j$ has child $w$, assume $m = \pi^l_i + \sum_{q \in C_i} \pi^l_q$ and $v^{get}_i = v_i(\pi^l_i) + \sum_{q \in C_i} v_q(\pi^l_q)$. Let $v_S^1, \cdots, v_S^{\mathcal{K}}$ be the top $\mathcal{K}$ ranked values in set $S$. Then $u_i = \mathcal{SW}_{-\mathcal{R_{\emph{l}}}} - \mathcal{SW}_{-D_i} = v_i^{get} - \sum_{k = \mathcal{K} - m + 1}^{\mathcal{K}} v_{Q \setminus D_i}^k$.

After $j$ has child $w$, if $C_j^{out}$ remains the same, $u_i$ does not change. If $C_j^{out}$ has one more buyer,  $\sum_{k = \mathcal{K} - m + 1}^{\mathcal{K}} v_{Q \setminus D_i}^k$ does not decrease and $m$ may decrease or remain the same. If $m$ decreases,  $u_i$ will decrease because every value in $v_i^{get}$ is larger than any value in $\sum_{k = \mathcal{K} - m + 1}^{\mathcal{K}} v_{Q \setminus D_i}^k$. If $m$ remains the same, $v^{get}_i$ remains the same, $u_i$ will not increase.

We can see $u_i$ is non-increasing when $j$ has more children. Thus, LDM is incentive compatible. 
\end{proof}

\end{appendix}

\end{document}